\documentclass[review]{elsarticle}

\pdfoutput = 1  

\usepackage{lineno}
\usepackage{hyperref}

\modulolinenumbers[1]
\usepackage{url}
\usepackage{amsfonts}
\usepackage{amsmath}
\usepackage{ntheorem}
\usepackage{graphicx, float}
\usepackage{caption}
\usepackage{subcaption}
\usepackage[toc, page]{appendix}
\usepackage{enumerate}
\usepackage[a4paper, right = 1.3in, top = 1.3in, left = 1.35in]{geometry}
\usepackage{xcolor}
\usepackage[final]{changes}
\newcommand{\stkout}[1]{\ifmmode\text{\sout{\ensuremath{#1}}}\else\sout{#1}\fi}
\setdeletedmarkup{\stkout{#1}}

\journal{Discrete Optimization 
}

\newcommand{\sgn}{\text{sgn}}
\newcommand{\relu}[1]{\big[#1\big]_+}
\newcommand{\ol}{\overline}

\newcommand{\new}{\text{new}}
\newcommand{\poly}{\text{poly}}
\def\Rbb{\mathbb{R}}
\def\0{\textbf{0}}
\def\eps{\varepsilon}

\newtheorem{theorem}{Theorem}[section]
\newtheorem{observation}{Observation}
\newtheorem{remark}[theorem]{Remark}

\newtheorem{proposition}[theorem]{Proposition}
\newtheorem{lemma}[theorem]{Lemma}

\newtheorem{corollary}[theorem]{Corollary}
\newtheorem{definition}{Definition}[section]

\newenvironment{proof}{\noindent{\bf Proof.}}{\hfill$\square$\medskip \\ }
\newenvironment{proof*}[1]{\noindent{\bf Proof of #1.}}{\hfill$\square$\medskip\\}

\let\oldequation\equation
\let\oldendequation\endequation

\renewenvironment{equation}
{\linenomathNonumbers\oldequation}
{\oldendequation\endlinenomath}
\begin{document}

\begin{frontmatter}

\title{Complexity of Training ReLU Neural Network}

\author{Digvijay Boob, Santanu S. Dey, Guanghui Lan}
\address{Industrial and Systems Engineering, Georgia Institute of Technology}

\begin{abstract}
In this paper, we explore some basic questions on the complexity of training neural networks with ReLU activation function. We show that it is NP-hard to train a two-hidden layer feedforward ReLU neural network. If dimension of the input data and the network topology is fixed, then we show that there exists a polynomial time algorithm for the same training problem. We also show  that if sufficient over-parameterization is provided in the first hidden layer of ReLU neural network, then there is a polynomial time algorithm which finds weights such that output of the over-parameterized ReLU neural network matches with the output of the given data.
\end{abstract}

\begin{keyword}
NP-hardness\sep Neural Network \sep ReLU Activation Function \sep 2-hyperplane Separability
\end{keyword}

\end{frontmatter}


\section{Introduction} 

Deep neural networks (DNNs) are functions computed on a graph parameterized by its edge weights. More formally, the graph corresponding to a DNN is defined by input and output dimensions $w_0, w_{k} \in \mathbb{Z}_{+}$, the number of hidden layers $k \in \mathbb{Z}_{+}$, and a sequence of $k$ natural numbers $w_1, w_2, \dots, w_{k}$ representing the number of nodes in each of the hidden $k$-layers. The function computed on a DNN graph is:

\[f := 
\tau \circ  a_{k}\circ \dots \circ a_2 \circ \tau \circ a_1,\]
where $\circ$ is function composition, $\tau$ is a nonlinear function (applied componentwise) called as the activation function, and $a_i: \mathbb{R}^{w_{i - 1}} \rightarrow \mathbb{R}^{w_{i}}$ are affine functions. 

Given the input and corresponding output data, the problem of training a deep neural network can be thought of as determining the edge weights of the directed layered graph that determine the affine functions $a_i$'s for which the output of the neural network matches the output data as closely as possible. Formally, given a set of input and output data $\{(x^i, y^i)\}_{i = 1}^N$ where $(x^i, y^i)\ \in \mathbb{R}^{w_0} \times  \mathbb{R}^{w_{k}}$, and a loss function $l: \mathbb{R}^{w_{k}}\times \mathbb{R}^{w_{k}}\rightarrow \mathbb{R}_{\ge 0}$ (for example, $l$ can be the square loss function), the task is to determine the weights that define the affine function  $a_i$'s such that 
\begin{equation}
\sum_{i = 1}^Nl(f(x^i),y^i) \label{eq:obj}
\end{equation}
is minimized.   

Some commonly studied activation functions are: threshold function, sigmoid function and ReLU function. ReLU is one of the most important activation functions used widely in applications. Despite its wide use, the question of computational complexity of training multi-layer fully-connected ReLU neural network has remained open. This paper makes contributions in this direction. Before formally stating our results, we first present the current understanding of the computational complexity in the literature.

\paragraph{Complexity of training DNNs with threshold activation function}  
The threshold (sign) function is given by 

\[\sgn(x) := \begin{cases}
1 &\text{if }x >0\\ -1 &\text{if }x<0 
\end{cases}.\]
It was shown by Blum et al.~\cite{BR88} that the problem of training a simple two layer neural network with two nodes in the first layer and one node in the second layer while using threshold activation function at all the nodes is NP-complete. The training problem turns out to be equivalent to separation by two hyperplanes problem which was shown to be NP-complete by Megiddo \cite{megiddo88}.  There are other hardness results such as crypto hardness for intersection of k-hyperplanes which apply to 
neural networks with threshold activation function.

\paragraph{DNNs with rectified linear unit (ReLU) activation function}

Theoretical worst case results presented above, along with limited empirical successes led to DNN's going out of favor by late 1990s. However, in recent times, DNNs became popular again due to the success of first-order gradient based heuristic algorithms for training. This success started with the work of \cite{HOT06} which presented empirical evidence that if DNNs are initialized properly, then we can find good solutions in reasonable runtime. This work was soon followed by series of early successes of deep learning in natural language processing \cite{CW08}, speech recognition \cite{MDH12} and visual object classification \cite{KSH12}. It was empirically shown in \cite{CSMBO16} that a sufficiently over-parameterized neural network can be trained to global optimality. 

These gradient-based heuristics are not useful for 
neural networks 
with threshold activation function 
as there is no gradient information. Even networks with sigmoid activation function fell out of favor because gradient information is not valuable when input values are large \cite{Hochreiter01gradientflow}. The popular neural network architecture uses \emph{ReLU activations} on which the gradient based methods are useful. Formally, the ReLU function is given by: $[x]_+ := \max(x,0)$.

\paragraph{Related literature}
As discussed before, most hardness results so far are for 
neural networks 
with threshold activation function 
\cite{BR88, KS09, SB14}. There are also  limited results for ReLU that we discuss next: Recently, \cite{LSS14} examined ReLU activation from the point of view that if two connected ReLU nodes are appropriately designed,  then it yields an approximation to the threshold function. Hence training problem for such a class of ReLU network should be as hard as 
training a 
neural network 
with threshold activation function
. Similar results are shown by \cite{DSS94}. In both these papers, in order to 
approximate the threshold activation function
, the neural network studied is not a fully connected network.  More specifically, in the underlying graph of such a neural network, each node in the second hidden layer is connected to exactly one distinct node in the first hidden layer, weight of the connecting edge is set to $-1$ with the addition of some positive bias term. Figure \ref{livni_vs_full} shows the difference between ReLU network studied by \cite{LSS14, DSS94} and fully connected ReLU network. The architecture artificially restricts the form of the affine functions in order to prove NP-hardness. In particular, it requires connecting hidden layer matrix to be a square diagonal matrix. Due to this restriction, it was unclear whether allowing non-diagonal entries of the matrix to be non-zero would make problem easier (more parameters hence higher power to neural network function) or more challenging (more parameters so more things to decide).

\begin{figure}[t]
	\begin{subfigure}{0.5\textwidth}
		\centering
		\includegraphics[width = 0.4\linewidth ]{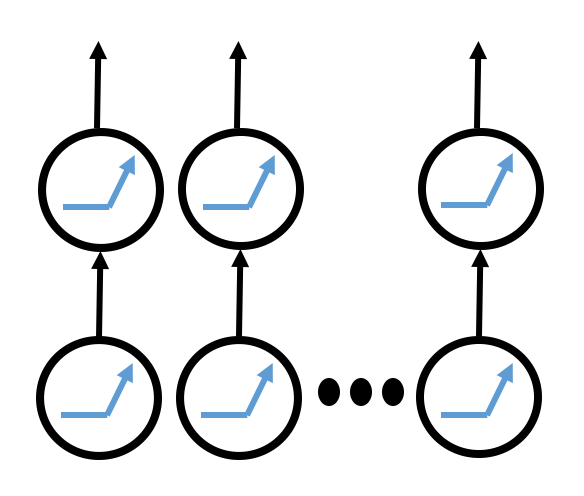}
		\caption{ReLU network studied in \cite{LSS14, DSS94}}
	\end{subfigure}
	\begin{subfigure}{0.5\textwidth}
		\centering
		\includegraphics[width=0.4\linewidth]{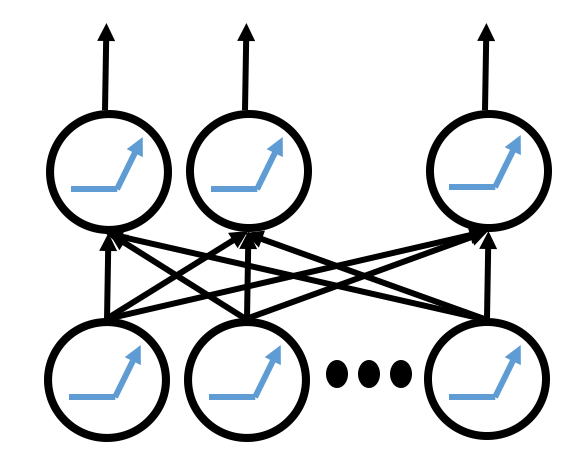}
		\caption{Fully connected ReLU}
	\end{subfigure}
	\caption{Difference between ReLU model studied in \cite{LSS14, DSS94} and typical fully connected counterpart}
	\label{livni_vs_full}
\end{figure}

Another line of research on the hardness of training ReLU neural networks assumes that the data is coming from some distribution. More recent works in this direction include \cite{S16} which shows a smooth family of functions for which the gradient of squared error function is not informative while training neural network over Gaussian input distribution. Another study in this line of work considers Statistical Query (SQ) framework \cite{SVWX17} (which contains SGD algorithms) and shows that there exists a class of  special functions generated by single hidden layer neural network for which learning will require exponential number of queries (i.e. sample gradient evaluations) for the data coming from the product measure of the real valued log-concave distribution. These are interesting studies in their own right and generally consider hardness with respect to the algorithms that use stochastic gradient queries and require that such algorithm must perform minimization of the (expectation) objective functions. In comparison, we consider the framework of NP-hardness which takes into account the complete class of the polynomial time algorithms, generally assumes that the data is given and requires an optimal solution to the corresponding empirical objective.

Recently, \cite{ABMM16} showed that a single hidden layer ReLU network can be trained in polynomial time when dimension of input, $w_0$, is constant. 

Based on the above discussion, we see that the status of the complexity of training the multi-layer fully-connected ReLU neural network remains open. Given the importance of  the ReLU NN, this is an important question. In this paper, we take the first steps in resolving this question.

\paragraph{Main Contributions}

\begin{itemize}
	\item NP-hardness: We show that the training problem for a simple two hidden layer fully-connected NN which has two nodes in the first layer, one node in the second layer and ReLU activation function at all nodes is NP-hard (Theorem \ref{thm:nphard}). Underlying graph of this network is exactly the same as that in Blum et al. \cite{BR88} but all activation functions are ReLU instead of threshold function.
	Techniques used in the proof are different from earlier work in the literature because there is no combinatorial interpretation to ReLU as opposed to the threshold function.
	\item Polynomial-time solvable cases: We present two cases where the training problem with ReLU activation function can be solved in polynomial-time. The first case is when the dimension of the input is fixed (Theorem \ref{lemma_poly_d}). This result generalizes the result from \cite{ABMM16} and uses the hyperplane arrangement theorem for its proof.

We also observe that when the number of nodes in the first layer of the network is equal to the number of input data points (Proposition \ref{lemma_N_node}) then there exists a polynomial time algorithm. The proof of this fact follows from a simple observation that reduces the problem to fitting a single hidden layer neural network and then applying the polynomial time algorithm result for single hidden layer neural network in the work of \cite{CSMBO16}. This is the highly over-parameterized neural network setting. This result leads to some interesting open questions that we discuss later. 
\end{itemize}

\section{Notation and Definitions}

We use the following standard set notation $[n]:=\{1,\dots, n\}$. 
The letter $d$ generally denotes the dimension of input data, $N$ denotes the number of data-points and unless explicitly specified, the output data is one dimensional.

The main training problem of interest for the paper corresponds to a neural network with 3 nodes. The underlying graph is a layered directed graph with two layers. The first layer contains two nodes and the second layer contains one node. The network is fully connected feedforward network. One can write the function corresponding to this neural network as follows:
\begin{equation}\label{2relu_nn_def}
F(x) = 
\relu{w_0+ w_1\relu{a_1(x)} + w_2\relu{a_2(x)}},
\end{equation}
\begin{figure}[h]
	\centering
	\includegraphics[scale = 0.3]{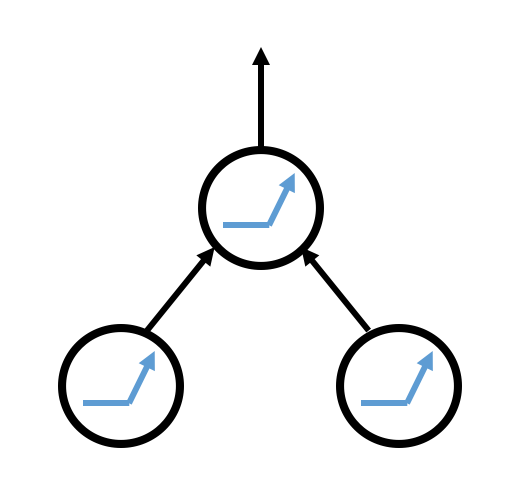}
	\caption{(2,1)-ReLU Neural Network. Also called 2-ReLU NN after dropping `1'. Here ReLU function is presented in each node to specify the type of activation function at the output of each node.}
	\label{2Relu}
\end{figure}
where $a_i: \mathbb{R}^d \rightarrow \mathbb{R}$ for $i \in \{1, 2\}$ are real valued affine functions, and $w_0, w_1, w_2 \in \mathbb{R}$. The output of the two affine maps $a_1, a_2$ are the inputs to the two ReLU nodes in first hidden layer of network. The weights $\{w_0,w_1,w_2\}$ denote affine map for ReLU node in second layer. 

We refer to the network defined in (\ref{2relu_nn_def}) as (2,1)-ReLU Neural Network(NN). As its name suggests, it has 2 ReLU nodes in first layer and 1 ReLU node in second layer. We will refer to $(k,j)$-ReLU NN as a generalization of $(2,1)$-ReLU NN where there are $k$ ReLU nodes in first layer and $j$ ReLU nodes in second layer. Note that the output of $(k,j)$-ReLU NN lies in $\Rbb^j$. If there is only one node in the second layer, we will often drop the $``1"$ and refer it as a 2-ReLU NN or k-ReLU NN depending on whether there are $2$ or $k$ nodes in the first layer, respectively.
Figure \ref{2Relu} shows $2$-ReLU NN. 
\begin{observation}\label{obs:1}
Note that 
\[w[ax+b]_+ \equiv \sgn(w)[|w|(ax+b)]_+ = \sgn(w)[\tilde{a}x + \tilde{b}],\] so without loss of generality we will assume $
w_1, w_2 \in \{-1,1\}$ in (\ref{2relu_nn_def}).
\end{observation}
Now we formally state the definition of the decision version of the training problem.
\begin{definition}[Decision-version of the  training problem]\label{defn:train} Given a set of training data $(x^i, y^i) \in \mathbb{R}^d \times \{1, 0\}$ for $i \in S$, do there exist edge weights so that the resulting function $F$ satisfies $F(x^i) = y^i$ for \added{all} $i \in S$.
\end{definition}
The decision version of the training problem in Definition \ref{defn:train} is asking if it is possible to find edge weights to obtain zero loss function value in the expression (\ref{eq:obj}), assuming $l$ is a norm, i.e., $l(a,b) = 0$ iff $a = b$. 

\section{Main Results}

\begin{theorem}\label{thm:nphard}
	It is NP-hard to solve the training problem for 2-ReLU NN.
\end{theorem} 
An immediate corollary of Theorem \ref{thm:nphard} is the following:
\begin{corollary}\label{cor:nphard}Training problem of (2,j)-ReLU NN is NP hard, for all $j \ge 1$.
\end{corollary}
The proof of Theorem~\ref{thm:nphard} is obtained by reducing the 2-Hyperplane Separability Problem to the training problem of 2-ReLU NN. Details of this reduction and the proof of Theorem \ref{thm:nphard} and Corollary \ref{cor:nphard} are presented in Section~\ref{sec:nphard}. 

While this paper was \replaced{in submission}{submitted}, two more studies \cite{manurangsi2018computational, dey2018approximation} considered the computational complexity of training a single ReLU node and proved that it is a NP-hard problem. \cite{manurangsi2018computational} also showed that it is NP-hard to train one hidden layer neural network with two nodes and ReLU activation at each node. This network basically removes the second layer ReLU activation and affine constant $w_0$ in \eqref{2relu_nn_def} so that neural network function of their case can be rewritten as $F(x) = w_1\relu{a_1(x)} + w_2\relu{a_2(x)}$. These are different network architectures and hence hardness of training any one of them does not necessarily imply hardness of training for remaining neural networks.

Megiddo \cite{megiddo88} shows that  the separability with fixed number of hyperplanes (generalization of  2-hyperplane separability problem) can be solved in polynomial-time in fixed dimension. Therefore 2-hyperplane separability problem can be solved in polynomial time given dimension is constant. Based on the reduction used to prove Theorem~\ref{thm:nphard} , a natural question to ask is ``Can one solve the training problem of 2-ReLU NN problem in polynomial time under the same assumption?". We answer this question in the affirmative.
\begin{theorem}\label{lemma_poly_d} Under the assumption that the dimension of input, $d$ and the number of nodes in the first layer, $k$, are constant, then there exists a poly(N)-time solution to the training problem of k-ReLU neural network, where $N$ is the number of data-points.
\end{theorem}

The high-level idea of the proof is the following: each data point ``passes through" the three ReLU nodes and the activation function in these nodes is ``turned on" or ``turned off" (i.e., the output is $0$ or not). We will enumerate all possible combinations of the data points being turned on or not, which we show is poly(N) assuming $d$ and $k$ is fixed (by use of the Hyperplane Arrangement Theorem). Then we show that for each of these combinations and for each possible sign pattern of the weights defining the affine function applied at the second layer, corresponding optimal affine functions can be calculated via solving one convex program of poly size. Finally,  we select the best optimal affine function which minimizes the loss function. 
Technique of Hyperplane Arrangement Theorem to enumerate partition was used in \cite{ABMM16} for proving poly(N)-time algorithms for single hidden layer neural networks. We extend this result for $k$-ReLU neural network which is a two hidden layer network. The complication due to second layer ReLU node are handled by solving a convex program of poly size. We present the formal proof of Theorem \ref{lemma_poly_d} in Appendix \ref{apx_lemma_poly_d}
.


We also study this problem under over-parameterization. Structural understanding of 2-ReLU NN yields an easy algorithm to solve training problem for N-ReLU neural network over N data points. In fact, the problem can be easily reduced to a single hidden layer NN.

\begin{proposition} \label{lemma_N_node} Given data, $\{x^i, y^i\}_{i \in [N]}$ (where we assume that $x^i$s are distinct), then the training problem for $N$-ReLU NN has a poly(N,d)-time randomized algorithm, where $N$ is the number of data-points and $d$ is the dimension of input.
\end{proposition}
The proof of this proposition is by first reducing the problem to that of training a single hidden layer network with $N$ nodes on a dataset of size $N$. Then a polynomial time algorithm from \cite{CSMBO16} is applied for interpolating the data. The precise details are presented in Appendix \ref{apx_lemma_N_node}.

\section{Training 2-ReLU NN is NP-hard}\label{sec:nphard}
In this section, we give details about the NP-hardness reduction for the training problem of 2-ReLU NN.  We begin with the formal definition of 2-Hyperplane Separability Problem.
\begin{definition}[2-Hyperplane Separability Problem]\label{2-aff-def} Given a set of points $\{x^i\}_{i \in [N]} \in \Rbb^d$ and a partition of $[N]$ into two sets: $S_1,S_0$, (i.e., $S_1 \cap S_0 = \emptyset$, $S_1 \cup S_0 = [N]$) 
	decide whether there exists two hyperplanes $H_1 = \{x: \alpha_1^Tx + \beta_1 =0\}$ and $H_2 = \{x: \alpha_2^Tx+\beta_2 = 0\}$ where $\alpha_1, \alpha_2 \in \Rbb^d$ and $\beta_1, \beta_2 \in \Rbb$ that separate the set of points in the following fashion:
	\begin{enumerate}[i]
		\item \label{part1} For each point $x^i$ such that $i \in S_1$, both $\alpha_1^Tx^i + \beta_1 > 0$ and $\alpha_2^Tx^i+\beta_2 > 0$.
		\item \label{part2}For each point $x^i$ such that $i \in S_0$, 
		$\alpha_1^Tx^i + \beta_1 < 0$ or $\alpha_2^Tx^i+\beta_2 < 0$.
	\end{enumerate}
\end{definition}
The 2-hyperplane separability problem is NP-complete \cite{megiddo88}. Note the difference between conditions \ref{part1} and \ref{part2} above. First one is an ``AND" statement and second is an ``OR" statement. Geometrically, solving 2-hyperplane separability problem means that finding two affine hyperplanes 
$\{\alpha_1,\beta_1\}$ and $\{\alpha_2,\beta_2\}$ such that all points in set $S_1$ lie in one quadrant formed by two hyperplanes and all points in set $S_0$ lie outside that quadrant. Due to this geometric intuition, the problem is called 2-hyperplane separability. We will construct a polynomial reduction from this NP-complete problem to
training 2-ReLU NN, which will prove that training 2-ReLU NN is NP-hard.

\begin{remark} [Variants of 2-hyperplane separability]\label{rem:variants}
Note here that some sources also define 2-hyperplane separability problem with a minor difference. This difference is that strict inequalities, `$>$', in Definition \ref{2-aff-def}.\ref{part1} are diluted to inequalities, `$\ge$'. In fact, these two problems are equivalent in the sense that there is a solution for the first problem if and only if there is a solution for the second problem. Solution for the first problem implies solution for the second problem trivially. Suppose there is a solution for the second problem, that implies there exist $\{\alpha_1,\beta_1\}$ and $\{\alpha_2, \beta_2\}$ such that for all $i \in S_0$ we have either $\alpha_1^Tx^i + \beta_1 < 0$ or $\alpha_2^Tx^i+\beta_2 < 0$. This implies $\epsilon := \min\limits_{i \in S_0} \max\{-\alpha_1x^i -\beta_1, -\alpha_2x^i -\beta_2\} > 0$. 
So if we shift both planes by $\frac{1}{2}\epsilon$ i.e. $\beta_i \leftarrow \beta_i+\frac{1}{2}\epsilon$ then this is a solution to the first problem.\end{remark}

\paragraph{Assumption: $\textbf{0} \in S_1$} (Here $\textbf{0} \in \Rbb^d$ is a vector of zeros.) Suppose we are given a generic instance of 2-hyperplane separability problem with data-points $\{x^i\}_{i \in [N]}$ from $\Rbb^d$ and partition $S_1$ and $S_0$ of the set $[N]$. Since the answer of 2-hyperplane separability instance is invariant under coordinate translation, we can shift the origin to any $x^i$ for $i \in S_1$, and therefore assume that the origin belongs to $S_1$ henceforth.

\subsection{Reduction}
Now we create a particular instance for 2-ReLU NN problem from a general instance of 2-hyperplane separability. We add two new dimensions to each data-point $x^i$. We also create a label, $y^i$, for each data-point. Moreover, we add a constant number of extra points to the training problem. Exact details are as follows:\\
Consider training set $\{(x^i, 0,0), y^i\}_{i\in [N]}$ where $y^i  = \begin{cases}
1 &\text{if } i \in S_1\\ 0 &\text{if } i \in S_0
\end{cases}$.

Add additional \replaced{18}{12} data points to the above training set as follows:\\
$\{p_1\equiv \{(\textbf{0},1,1),1\}, p_2 \equiv \{(\textbf{0},1.5, 0.75),1\}, p_3 \equiv \{(\textbf{0},2,1),1\}, p_4 \equiv \{(\textbf{0},2.25,1.5),1\},$ \\
$p_5~\equiv~\{(\textbf{0},2,2),1\}, p_6 \equiv \{(\textbf{0},1.5,2.25),1\}, p_7 \equiv \{(\textbf{0},1,2),1\}, p_8 \equiv \{(\textbf{0},0.75,1.5),1\},$\\
$p_{9} \equiv \{(\textbf{0},0,-1),0\}, p_{10} \equiv \{(\textbf{0},1,-1),0\}, p_{11} \equiv \{(\textbf{0},2,-1),0\},$\\
$p_{12} \equiv \{(\textbf{0},3,-1),0\}, p_{13} \equiv \{(\textbf{0},5,-1),0\},$\\
$ p_{14} \equiv \{(\textbf{0},-1,0),0\}, p_{15} \equiv \{(\textbf{0},-1,1),0\}, p_{16} \equiv \{(\textbf{0},-1,2),0\},$\\
$ p_{17} \equiv \{(\textbf{0},-1,3),0\}, p_{18} \equiv \{(\textbf{0},-1,5),0\}\}.$\\
We call the set of additional data points with label $1$ as $T_1$ and the set of additional data points with label $0$ as $T_0$. These additional data points (we refer to these points as the
\begin{figure}[h]
	\centering
	\includegraphics[scale = 0.37]{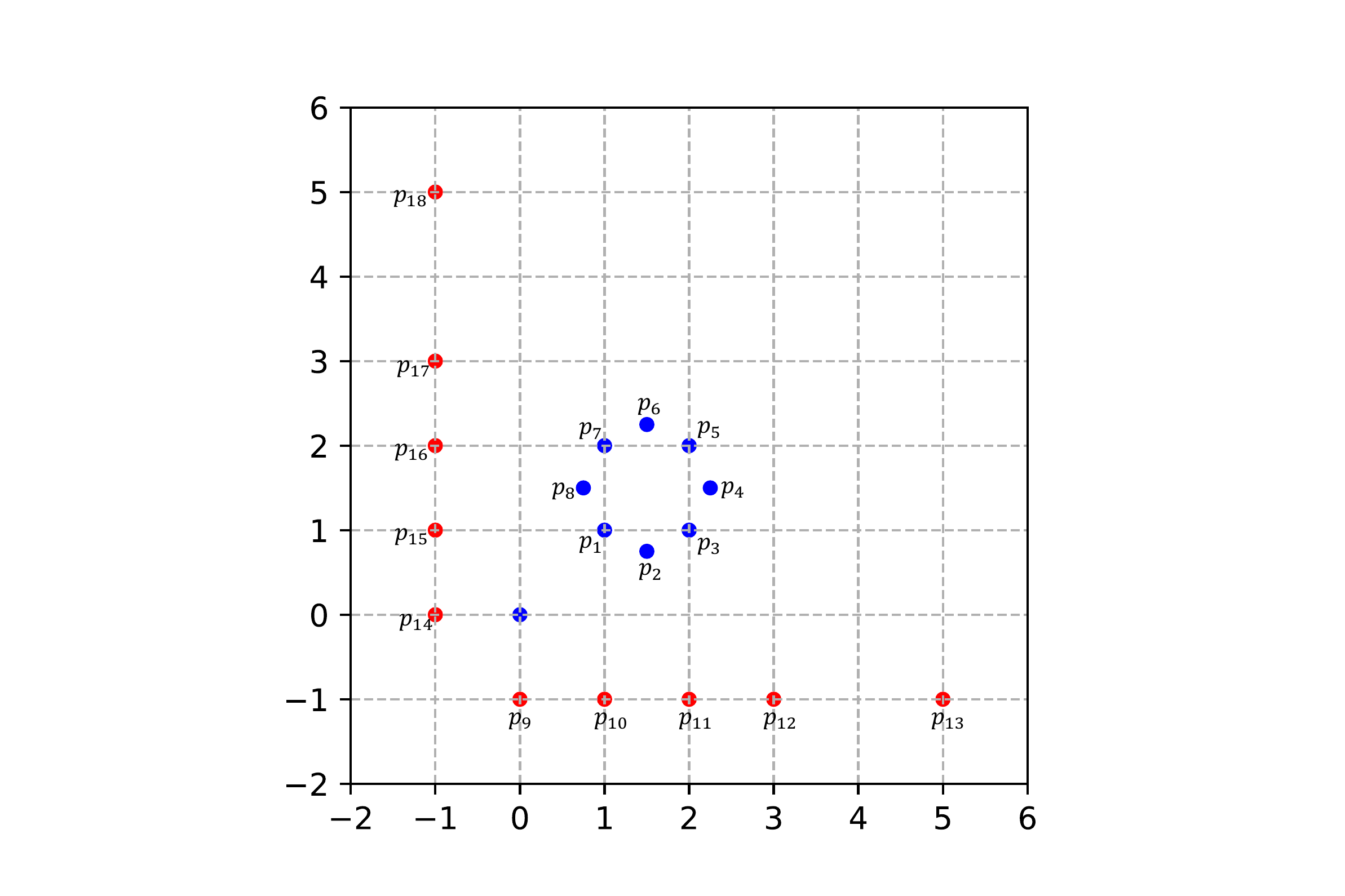}
	\caption{Gadget: Blue points represent set $T_1$ and red points represent set $T_0$. 
	}
	\label{gadget}
\end{figure}
``gadget points") are of fixed size. So this is a polynomial time reduction.

Figure \ref{gadget} shows the gadget points. Note that origin is added to the gadget because there exists $i \in S_1$ such that $x^i = \textbf{0}$. Hence training set has the data-point $\{(\textbf{0}, 0,0), 1\}$.

Henceforth we refer to the training problem of fitting 2-ReLU NN to this data as (\textbf{P}). In the context of the training problem (\textbf{P}), we abuse the notation and call the set of points $(x^i, 0, 0)$ with label $1$ as $S_1$ and the set of points $(x^i, 0, 0)$ with label $0$ as $S_0$. In particular, there is a direct correspondence between the sets $S_1, S_0$ defined in 2-hyperplane separability problem and sets $S_1, S_0$ defined for 2-ReLU NN training problem (\textbf{P}). Use of our notation is generally clear from the context.

Now what remains is to show that the general instance of 2-hyperplane separability has a solution if and only if the constructed instance of 2-ReLU NN has a solution. In order to understand our approach 
better
, we introduce the notion of ``hard-sorting". Hard-sorting is formally defined below, and its significance is stated in Lemma \ref{lemma_hard-sort}.
\begin{definition}[Hard-sorting] \label{def:hard-sort}
	We say that a set of points $\{\pi^i\}_{i \in S}$, partitioned into two sets $\Pi_0, \Pi_1$ can be {hard-sorted with respect to $\Pi_1$} if there exist two affine transformations $l_1,l_2$ and scalars $w_1, w_2, c$ such that 
	the following 
	condition 
	is 
	satisfied: 
	\begin{equation}\label{eq:hard-sort}
	w_1\relu{l_1(\pi)} + w_2\relu{l_2(\pi)}  
	\begin{cases}
	= c &\text{for all } \pi \in \Pi_1 \\ 
	< c &\text{for all } \pi \in \Pi_0 
	\end{cases}
	\end{equation}
\end{definition}
Being able to hard-sort implies that after passing the data through two nodes of the first hidden layer, the scalar input to the second hidden layer node must have a separation of the data-points in $\Pi_1$ and the data-points in $\Pi_0$. Moreover, scalar input corresponding to all data points in $\Pi_1$ must be equal. 
\begin{remark} \label{parity_invariant_remark}If there exists scalars $w_1, w_2, c$ and affine transformations $l_1, l_2$ such that \[w_1\relu{l_1(\pi)} + w_2\relu{l_2(\pi)}  \begin{cases}
	 = c &\text{for all }\pi \in \Pi_1;\\
	 > c &\text{for all }\pi \in \Pi_0,
	\end{cases}\] then $-w_1,-w_2, -c, l_1, l_2$ satisfy condition \eqref{eq:hard-sort} of hard-sorting.
\end{remark}
\begin{remark} \label{subset_prop_hard-sort_remark} Let $\ol{\Pi}_0 \subset \Pi_0$ and $\ol{\Pi}_1 \subset \Pi_1$. Then hard-sorting of $\Pi_0 \cup \Pi_1$ with respect to $\Pi_1 \Rightarrow$ hard-sorting of $\ol{\Pi}_0 \cup \ol{\Pi}_1$ with respect to $\ol{\Pi}_1$.
\end{remark}
\begin{remark}\label{rem:w_pm1}
	Without loss of generality, we may assume that $w_1, w_2 \in  \{-1,1\}$.
\end{remark}
It is not difficult to see that hard-sorting implies (\textbf{P}) has a solution. We show that hard-sorting is also required for solving training problem. This is formally stated in lemma below.
\begin{lemma}\label{lemma_hard-sort} The 2-ReLU NN training problem (\textbf{P}) has a solution if and only if data-points $S_1 \cup T_1 \cup S_0 \cup T_0$ are hard-sorted with respect to $S_1 \cup T_1$. \end{lemma}
The proof of Lemma \ref{lemma_hard-sort} can be found in  Appendix~\ref{sec:proof_hard-sort} . Figure \ref{hard_sort} presents a geometric interpretation of Lemma \ref{lemma_hard-sort}
\begin{figure}[H]
	\begin{subfigure}[t]{0.32\textwidth}
		\includegraphics[width = 0.95\linewidth]{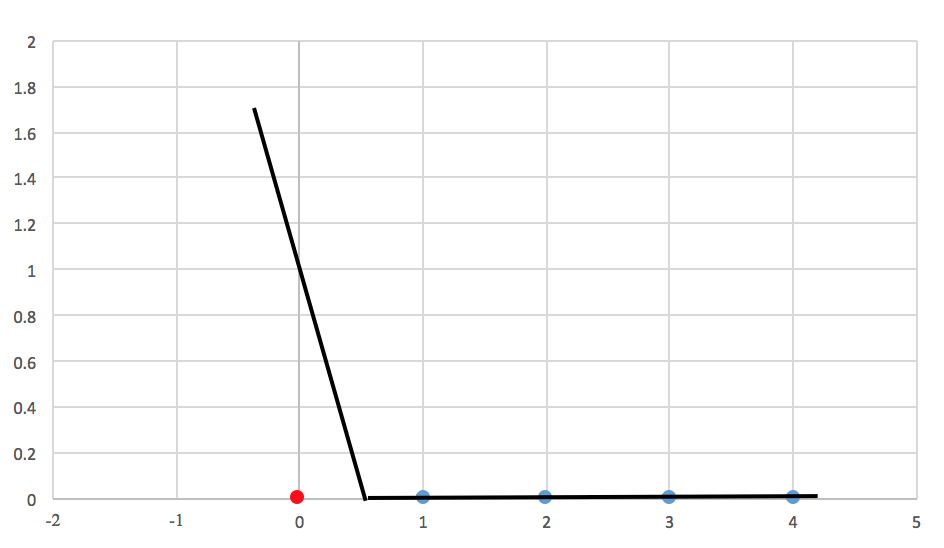}
		\caption{Input is hard-sorted. This can give a perfect fit.}
	\end{subfigure}\hfill
	\begin{subfigure}[t]{0.32\textwidth}
		\includegraphics[width = 0.95\linewidth]{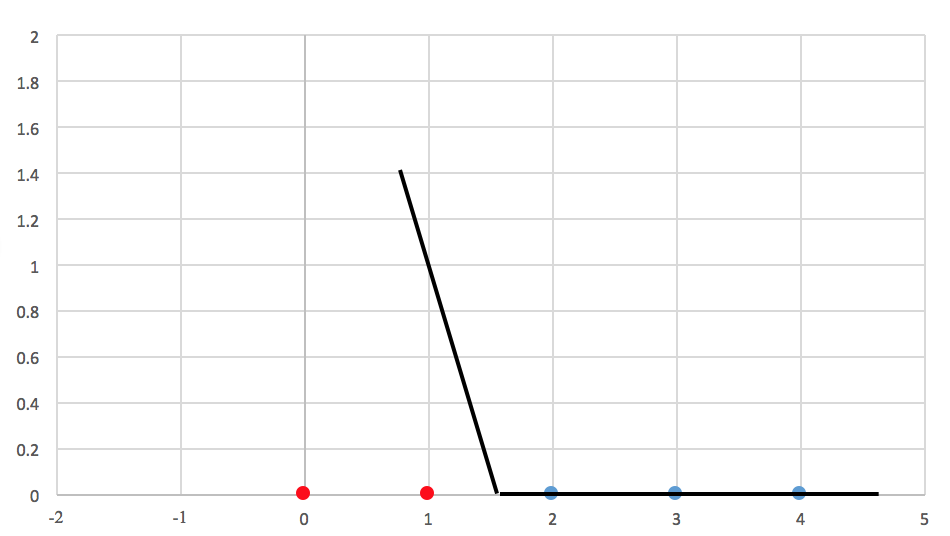}
		\caption{Since there are two red points so input is not hard-sorted. This cannot give a perfect fit.}
	\end{subfigure}\hfill
	\begin{subfigure}[t]{0.32\textwidth}
		\includegraphics[width = 0.95\linewidth]{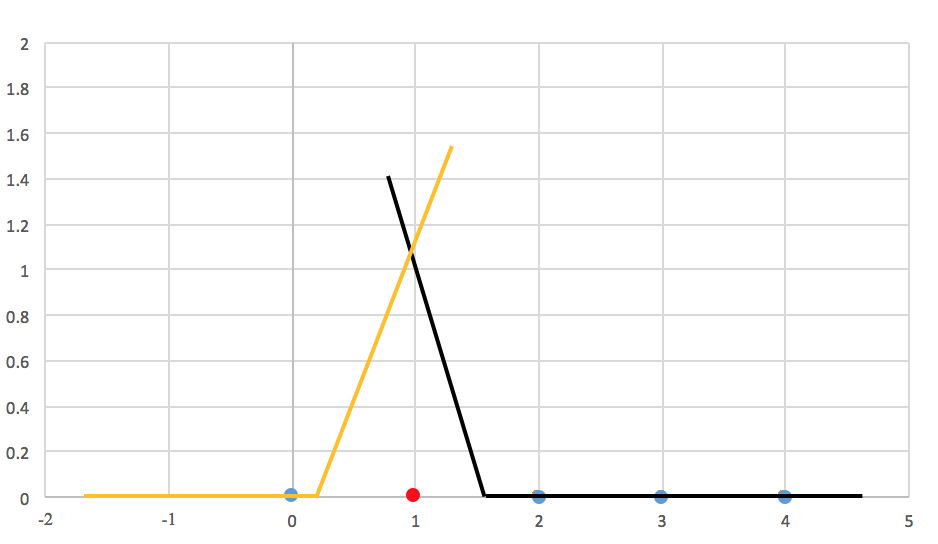}
		\caption{Since blue points lies on different side of red points so input is not hard-sorted. This cannot give a perfect fit.}
	\end{subfigure}
	\caption{X-axis in figures above is output of the first layer of 2-ReLU NN i.e. $w_1\relu{l_1(\pi)}+w_2\relu{l_2(\pi)}$. Y-axis is the output of second hidden layer node. Since output of first hidden layer goes to input of second hidden layer, we are essentially trying to fit ReLU node of second hidden layer. In particular, red and blue dots represent output of first hidden layer on data points with label 1 and 0 respectively. In Fig. (a) we see that hard-sorted input can be classified as $0/1$ by a ReLU function. In Fig. (b) and Fig. (c) we see that input which is not hard-sorted cannot be classified exactly as $0/1$ by a ReLU function.}
	\label{hard_sort}
\end{figure}

We use the hard-sorting characterization of the solution of the training problem (\textbf{P}) extensively. We first show the forward direction of the reduction in the lemma below. This is also the easier direction.
\begin{lemma}\label{forward_direction}
	If 2-hyperplane separability problem has a solution, then problem (\textbf{P}) has a solution.
\end{lemma}
The proof of Lemma \ref{forward_direction} can be found in Appendix~\ref{sec:proof_forward_direction}.\\
To prove \added{the} reverse direction we need to show that if a set of weights solve the training problem (\textbf{P}) then we can generate a solution to the 2-hyperplane separability problem. 
In the rest of the proof we will argue that the only way to solve the training problem (\textbf{P}) for 2-ReLU NN or equivalently hard-sort data-points is to find two affine function\added{s} $a_1, a_2:\Rbb^{d+2}\to \Rbb$ such that i) $a_1(x) \le 0$ and $a_2(x) \le 0$ for all $x \in S_1 \cup T_1$ and ii) $a_1(x) > 0$ or $a_2(x) > 0$ for all $x \in S_0\cup T_0$. If such a solution exists then there exists a solution to 2-hyperplane separability problem after dropping coefficients of last two dimensions of affine functions $-a_1$ and $-a_2$. 
Note that changing `$<$' to `$\le$' in 2-affine separability problem is valid in view of Remark \ref{rem:variants}.


We will first show that we can hard-sort the gadget points only under the properties of $a_1$ and $a_2$ mentioned above. This implies that a solution to (\textbf{P}) which hard-sorts all points (including the gadget points) must have same properties of $a_1$ and $a_2$. This follows from counter-positive of Remark \ref{subset_prop_hard-sort_remark}, i.e., if a subset of data-points cannot be hard-sorted then all data-points cannot be hard-sorted. Henceforth, we will focus on the gadget data-points (or the last two dimensions of the data). 
\subsubsection{Gadget Points and Hard-Sorting}
In the following lemma, we show a necessary condition on $a_1, a_2$ satisfying hard-sorting of gadget data points $T_1\cup T_0 \cup \{\0\}$ with respect to $T_1\cup \{\0\}$.
\begin{lemma}\label{final_lemma_new}
	If affine functions $a_1, a_2: \Rbb^{d+2} \to \Rbb$ and scalars $w_1, w_2, c$ satisfy hard-sorting of the data-points $T_1 \cup T_0 \cup \{\0\}$ with respect to $T_1 \cup \{\textbf{0}\}$, then all points in $T_1$ must satisfy $a_1(x) \le 0$, $a_2(x) \le 0$. Moreover, we must have $w_1 = w_2 = -1$ and $c = 0$.
\end{lemma}
Note that in view of Lemma \ref{final_lemma_new} and counter-positive of Remark \ref{subset_prop_hard-sort_remark}, we have that affine function $a_1, a_2:
\Rbb^{d+2}\to \Rbb$ and scalars $w_1, w_2,c$ satisfying hard-sorting of $S_1\cup T_1 \cup S_0\cup T_0$ with respect to $S_1\cup T_1$ must satisfy 
\[-\relu{a_1(x)} -\relu{a_2(x)} \begin{cases}
 = 0 &\text{if } x \in S_1;\\
 < 0 &\text{if } x \in S_0
\end{cases} .\]
The above condition is equivalent to the requirement that $a_1(x) \le 0,\ a_2(x) \le 0$ for all $x \in S_1$ and $a_1(x) > 0$ or $a_2(x) > 0$ for $x \in S_0$. After dropping the last two dimensions of $-a_1$ and $-a_2$, we obtain the solution for 2-affine separability problem. Now that we have reduced the problem to the key lemma above, the main purpose of this section is to prove Lemma \ref{final_lemma_new}.

Note that for each data point in the gadget $T_1 \cup T_0 \added{\cup} \{\0\}$, the first $d$ elements are always $0$. So for the sake of gadget, we may assume that $a_1, a_2: \Rbb^2 \to \Rbb$ and the gadgets lies in $\Rbb^2$. They can be thought of as the projection of the original $a_i:\Rbb^{d+2} \to \Rbb$  and $\textbf{0} \in \Rbb^{d+2}$ to last two dimension which are relevant for gadget data points 
$T_1 \cup T_0$. Due to this observation, we assume that $a_1,a_2: \Rbb^2\to \Rbb$ henceforth for this subsection, and provide a proof of Lemma \ref{final_lemma_new} under this assumption.

The proof of Lemma \ref{final_lemma_new} is divided into the following sequence of results.
\begin{proposition}\label{prop:simplified}
	Suppose that $a_1, a_2$ satisfy hard-sorting of $T_1\cup T_0$ with respect to $T_1$\replaced{. Then}{ then} there exists $x \in T_1$ such that $a_1(x) \le 0,\ a_2(x) \le 0$.
\end{proposition}
Proof of Proposition \ref{prop:simplified} can be found in Appendix \ref{apx_prop_simplified}.

Next we show one more simple proposition which is critical in proving the final result. The proof of this proposition can be found in Appendix \ref{apx_lemma_parity1}.
\begin{proposition} \label{lemma_parity1} If affine functions $a_1, a_2$ and weights $w_1, w_2$ satisfy hard-sorting of $T_1 \cup T_0 \cup \{\0\}$ with respect to $T_1 \cup \{\0\}$, then $w_1, w_2$ \textbf{must} satisfy $w_1 = w_2 =-1$.
\end{proposition}
We are now ready to present the prove Lemma \ref{final_lemma_new}.\\
\begin{proof*}{Lemma \ref{final_lemma_new}}
	Since $a_1, a_2$ satisfy hard-sorting of the data points $T_1\cup T_0\cup\{\0\}$ with respect to $T_1\cup \{\0\}$ then, in view of Proposition \ref{prop:simplified} and Proposition \ref{lemma_parity1}, we have 
	\begin{enumerate}
		\item $\exists x \in T_1 $ such that $a_1(x) \le 0, a_2(x) \le 0$.
		\item $w_1 =w_2 = -1$.
	\end{enumerate}
	Then we have that $-\relu{a_1(x)} -\relu{a_2(x)} = 0$ for all $x\in T_1$, due to condition \eqref{eq:hard-sort} of hard-sorting. This implies $a_1(x) \le 0, a_2(x)\le 0$ for all $x \in T_1$. So we conclude the proof.
\end{proof*}
In the next section, we show that this result on the gadget data-points gives us the solution to the original 2-hyperplane separability problem.
\subsubsection{From Gadget Data to Complete Data }
\begin{lemma}\label{reverse_direction} If there is a solution to the problem (\textbf{P}), then there is a solution to corresponding 2-hyperplane separability problem.
\end{lemma}
\begin{proof}
	Note that if there is a solution to problem (\textbf{P}), 
	then by Lemma \ref{lemma_hard-sort}, we must have $a_1, a_2: \Rbb^{d+2}\to \Rbb$ and $w_1, w_2, c$ hard-sorting $S_1\cup T_1 \cup S_0\cup T_0$ with respect to $S_1\cup T_1$. In view of Lemma \ref{final_lemma_new} and counter-positive of Remark \ref{subset_prop_hard-sort_remark}, we have
	\begin{enumerate}
		\item \label{final_part1} $w_1 =  w_2=-1$.
		\item \label{final_part2} $w_1\relu{a_1(x)} + w_2\relu{a_2(x)} = 0$ for all $x \in S_1 \cup T_1$ due to requirement \eqref{eq:hard-sort} of hard-sorting.
	\end{enumerate}
	Since 
	$w_1 =  w_2 = -1$, so \ref{final_part2} above implies $a_1(x) \le 0$ and $a_2(x)\le 0$ for all $x \in S_1\cup T_1$. Moreover, we require $a_1(x) > 0$ or $a_2(x) > 0$ for all $x \in S_0 \cup T_0$ due to condition \eqref{eq:hard-sort} of hard-sorting. Now as discussed earlier, $-a_1, -a_2$ after ignoring coefficients of last two dimensions will yield solution to 2-hyperplane separability problem. Hence we conclude the proof.
\end{proof}
Now we are ready to prove the main NP-hardness theorem.\\
\begin{proof*}{Theorem \ref{thm:nphard}}
	Using Lemma \ref{forward_direction} and Lemma \ref{reverse_direction}, we  conclude the proof.
\end{proof*}
Below we state an immediate corollary of Theorem \ref{thm:nphard} whose proof can be found in Appendix \ref{apx_cor_2j}.
\begin{corollary} \label{cor-2-j}
	Training problem of (2,j)-ReLU NN is NP hard.
\end{corollary}
\section{Discussion}
We showed that the problem of training $2$-ReLU NN is NP-hard. Given the importance of ReLU activation function in neural networks, in our opinion, this result resolves a significant gap in understanding complexity class of the problem at hand. On the other hand, we show that the problem of training $N$-ReLU NN is in P. So a natural research direction is to understand the complexity status when input layer has more than 2 nodes and strictly less than $N$ nodes. A particularly interesting question in that direction is to generalize the gadget we used for 2-ReLU NN to the case of k-ReLU NN.

\subsubsection*{Acknowledgments}
We would like to thank anonymous reviewers whose comments helped in simplifying several results in this paper.
\begin{appendices}
\section{Proofs of Auxiliary Results}
In this appendix, we provide proof of all auxiliary results.
\subsection{Proof of Theorem \ref{lemma_poly_d}}
\label{apx_lemma_poly_d}
	Suppose we partition the set $[N]$ into sets $Q_j$ and $\ol{Q}_j$ such that all points in $Q_j$ satisfy $a_j(x) \ge 0$ and all points in $\ol{Q}_j$ satisfy $a_j(x)< 0$ for all $j \in [k]$. Given a set $S \subseteq [k]$, we define 
	$T(S) := \left(\bigcap_{j \in S} Q_j\right) \cap \left(\bigcap_{j \in \bar{S}} \bar{Q}_j\right) 
	$
	where $\ol{S} = [k]\setminus S$. Let $z = (a_1,\dots, a_k, w_0,w_1,\dots, w_k
	$. Then the objective function can be written as 
	\begin{equation*}
	f(z) =\sum\limits_{S \subseteq [k]} \sum\limits_{i \in T(S)}\Big( 
	\relu{w_0 + \sum\limits_{j \in S}w_ja_j(x^i)}-y_i\Big)^2.
	\end{equation*}
	Now we can partition $T(S)$ into sets $T(S)_{1}$ and $T(S)_{2}$ for each $S \subseteq [k] , S \neq \phi$. For $T(S)_{1}$, the ReLU term in the objective, $w_0 + \sum\limits_{j \in S}w_ja_j(x)$ (note that this is an affine function), is constrained to be non-negative and for $T(S)_{2}$ the ReLU terms is constrained to be non-positive. We need not enumerate partitions of $T(\phi)$ since ReLU terms for $T(\phi)$ do not depend on data-points. The key observation is that the partition of $T(S)$ into sets $T(S)_1$ and $T(S)_2$ is a partition due to a hyperplane. 
	
	Number of combinations: According to the Hyperplane Arrangement Theorem, given a set of points $\{x^i\}_{i \in N}$ in $\Rbb^d$, the number of distinct partitions created by linear separators is $O(N^d)$. Moreover, due to~\cite{EOS86}, we can enumerate all possible partitions created by linear separators in $O(N^d)$ time. Therefore, there are a total of $O(N^{kd})$ possible combinations of $Q_j, j \in [k]$. For each such $Q_j, j \in [k]$, there are $2^{k}$ non-empty subsets $T(S) \subseteq [N]$. For each $T(S), S \ne \phi$, there are $O(|T(S)|^d) = O(N^d)$ possible ways to partition $T(S)$ into $T(S)_{1}$ and $T(S)_{2}$. So number of product combinations is $O(N^{(2^k-1)d})$. Hence there are a total of $O(N^{(kd + (2^k -1) d) })$ combinations.
	
	Number of convex programs: By Observation~\ref{obs:1} it suffices to check for $
	w_1,\dots,  w_k = \pm 1$. We will divide the optimization problem in two cases $w_0\ge 0$ and $w_0\le0$. 
	So there are a total of $2^{k+
		1}
		$ convex programs for each possible combination of $Q_j$, $T(S)_1$ for all $S \subseteq [k]$ of the following form:
	
	%
	
	\begin{equation*}
	\textup{min} \sum\limits_{\substack{S \subseteq [k], \\ S \neq \phi}} \Bigg\{ \sum\limits_{i \in T(S)_1} \Big(
	\Big(w_0 + \sum\limits_{j \in S} w_ja_j(x^i)\Big)-y_i\Big)^2 +\sum\limits_{i \in T(S)_2} (0 -y_i)^2 \Bigg\} + \sum\limits_{i \in T_\phi} \Big(
	\relu{w_0} - y_i \Big)^2
	\end{equation*}
	subject to constraints\\
	\begin{equation} \label{eq:int_rel1}
	\begin{alignedat}{2}
	a_j(x^i) &\ge 0, \quad \quad & &\forall j,\ i\in Q_j\\
	a_j(x^i) &\le 0, & &\forall j,\ i\in \ol{Q}_j
	\end{alignedat}
	\end{equation}
	\begin{equation}\label{eq:int_rel2}
	\begin{alignedat}{2}
	w_0 + \sum\limits_{j \in S}w_ja_j(x^i) &\ge 0,\quad \quad \quad \quad & &\forall S 
	\subseteq [k], S\neq \phi,\ i \in T(S)_{1}\\
	w_0 + \sum\limits_{j \in S}w_ja_j(x^i) &\le 0,& &\forall S \subseteq [k], S\neq \phi, \ i \in T(S)_{2}
	\end{alignedat}
	\end{equation}
	Moreover, we add constraint $w_0 \ge 0$ or $w_0\le 0$ and change the $\relu{w_0}$ term in objective with $w_0$ or $0$ respectively. Every program has $k(d+1)+1$ variables in $a_1, \dots, a_k, w_0$. Total number of constraints is at most $kN+N+1$. Note that, for constraints of type \eqref{eq:int_rel1}, for each $j$, number of constraints equals $|Q_j \cup \ol{Q}_j| = N$. Hence total number of constraints of type \eqref{eq:int_rel1} are $kN$. Similarly, for constraints of type \eqref{eq:int_rel2}, for each $S \subseteq [k]$, we have total of $|T(S)_1 \cup T(S)_2| = |T(S)|$ constraints. Hence total constraints of type \eqref{eq:int_rel2} are $\sum\limits_{\substack{S \subseteq [k], \\ S \neq \phi}} |T(S)| \le N$ (This follows due to observation that $T(S), S\subseteq [k]$ is a partition of $[N]$). One more constraint is on $w_0$. Hence total number of constraints is $(k+1)N+1$. Since number of constraints and variables are $\poly(k,d,N)$ and objective is convex quadratic so we conclude that this program can be solved in $\poly(N,k,d)$ time. 
	
	Finally, the total number of convex programs to be solved is $O(2^{k+1} \cdot N^{kd +(2^k-1)d})$.

\subsection{Proof of Proposition \ref{lemma_N_node}}\label{apx_lemma_N_node}
	Before proving this proposition, we state a polynomial time algorithm (Theorem 1 of \cite{CSMBO16}) for training single hidden layer neural network.
	\begin{proposition}\label{prop_N_node}
		 Let $f(x) = \sum_{j =1}^{N}w_j\relu{a_j(x)} + w_0$ be a single hidden layer neural network with $N$ nodes. Let $x^i \in \Rbb^d, i =1, \dots, N$ be distinct data-points, and let $y^i \in \Rbb, i = 1, \dots, N$ be arbitrary labels. Then there exists a poly($N,d$)-time algorithm to train this neural network such that $f(x^i) = y^i,  i =1, \dots, N$.
	\end{proposition}
	Now we are ready to prove Proposition \ref{lemma_N_node}.\\
	Note that a N-ReLU NN can be written as $c(x) = \relu{\sum\limits_{j=1}^N w_j\relu{a_j(x)} + w_0}$. Suppose $y = [y^1, \dots, y^N]^T \in\Rbb^N$ be a vector of labels. We may assume that $y \ge \0$ since otherwise \replaced{the answer to the training problem is a straight-forward ``No"}{we can add a constant term to each label in $y$}. \added{Now, we prove that weights $w_i,  i =0, \dots, N$ and affine functions $a_j, j = 1, \dots, N$, satisfying $c(x^i) = y^i$ for all $i \in [N]$, can be computed in poly($N,d$)-time. Hence, providing a ``Yes" answer to the training problem in poly($N,d$)-time.} 
	
	Since $y^i \ge \0$ for all $i \in [N]$ and we want $y^i = c(x^i) = \relu{f(x^i)}$ \deleted{so we have $f(x^i) = y^i$}, where $f(x) = \sum\limits_{j=1}^N w_j\relu{a_j(x)} + w_0$, \added{it suffices}\footnote{Note that when $y^i = 0$, we can have $f(x^i) \le 0$. However, assuming $f(x^i) = 0$ is sufficient since we can find a solution for arbitrary nonpositive assignments of $y^i$. This also keeps the proof clean as we can assume that $f(x^i) = y^i$ for all labels $y^i$ uniformly (even for the $0$-labels).} to show that $f(x^i) = y^i$ for all $i \in [N]$. 
	Using the fact that number of nodes in $f$ matches the number of data points, $N$, then applying Proposition \ref{prop_N_node}, we obtain the result.

\subsection{Proof of Lemma \ref{forward_direction}} \label{sec:proof_forward_direction}
Suppose $(\alpha_1, \beta_1)$ and $(\alpha_2, \beta_2)$ are solution satisfying condition for 2-hyperplane separability. Note that there is a data-point $\textbf{0} \in S_1$ so we obtain $\beta_1, \beta_2 > 0$. Without loss of generality we can assume $\beta_1 = \beta_2 = 0.5$. This is due to the fact that scaling the original solution by any positive scalar yields a valid solution. 
Now we show that the solution of 2-hyperplane separability problem can be used to show hard-sorting of $S_0\cup T_0 \cup S_1 \cup T_1$ with respect to $S_1\cup T_1$. Hence in view of Lemma \ref{lemma_hard-sort}, we obtain the existence of a solution for problem (\textbf{P})
. 

Set $w_1 = w_2 = -1,\ c = 0$. Moreover, for $(x, y, z)\in \Rbb^{d+2}$, consider the affine map $l_1(x,y,z) = -\alpha_1^Tx-y-\beta_1$ and $l_2(x,y,z) = -\alpha_2^Tx-z -\beta_2$. We claim that $w_1, w_2, c, l_1, l_2$ satisfy hard-sorting condition \eqref{eq:hard-sort} for $S_0\cup T_0 \cup S_1 \cup T_1$ with respect to $S_1 \cup T_1$. In particular, note that 

\begin{itemize}
	\item[1.] For $x \in S_1$, we have \[ -\relu{-\alpha_1^Tx-\beta_1} - \relu{-\alpha_2^Tx-\beta_2}  = 0 = c.\]
	\item[2.] For $x = (\0,l,m) \in T_1$, we have 
	\[-\relu{-\beta_1 -l } - \relu{ -\beta_2 -m} = 0.\] This follows since $\beta_1 = \beta_2 = 1/2$ and $l,m \in [0.75,2.25]$ so the two ReLU terms inside are both zero for all $x \in T_1$.
	\item[3.]For $x \in S_0$, we have \[ -\relu{-\alpha_1^Tx-\beta_1} - \relu{-\alpha_2^Tx -\beta_2} <0 .\] This follows since at least one of $\alpha_1^Tx +\beta_1$ and $\alpha_2^Tx+\beta_2$ is strictly negative for $x\in S_0$ as $(\alpha_1, \beta_1)$ and $(\alpha_2, \beta_2)$ are solution for 2-hyperplane separability problem.
	\item[4.] For $x = (\0, l,m) \in T_0$, we have 
	\[-\relu{-\beta_1 -l } - \relu{ -\beta_2 -m} < 0.\] This follows since $\beta_1 = \beta_2 = 1/2$ and either $l$ or $m$ equals $-1$ for $x \in T_0$.
\end{itemize}
This proves hard-sorting of $S_0\cup T_0 \cup S_1 \cup T_1$ with respect to $S_1 \cup T_1$ and hence we have the existence of solution for training problem (\textbf{P}).

\subsection{Proof of Lemma \ref{lemma_hard-sort}} \label{sec:proof_hard-sort}
We first prove the forward direction. 
Suppose points are hard-sorted as required by 
the 
lemma
. Then define $\eps := \min\limits_{x \in S_0 \cup T_0} -w_1 \relu{l_1(x)} - w_2\relu{l_2(x)} + c$. By definition, we have $\eps > 0$. Then 
neural network $f(x) = \frac{2}{\eps}\relu{w_1\relu{l_1(x)} +w_2\relu{l_2(x)} - c + \eps/2}$ solves training problem. 
This can be easily checked from the fact that
$$w_1\relu{l_1(x)} + w_2\relu{l_2(x)} -c \begin{cases}
	= 0 &\text{if } x \in S_1\cup T_1; \\ < -\eps &\text{if }x \in S_0\cup T_0,	\end{cases}$$ 
	which holds under the assumption of hard-sorting
	.
\\
Now we assume that points cannot be hard-sorted and conclude that there does not exist weight assignment solving training problem of 2-ReLU NN, hence proving the backward direction. Since the points cannot be hard-sorted so there does not exist any $l_1, l_2, w_1, w_2, c$ satisfying 
condition \eqref{eq:hard-sort}
. This fact along with Remark \ref{parity_invariant_remark} implies that for all possible weights we either have\\ 
a) $w_1 \relu{l_1(x)} + w_2\relu{l_2(x)}$ is not constant for all $x \in S_1 \cup T_1$ or\\
b) If $w_1\relu{l_1(x)} + w_2\relu{l_2(x)} = c$ for all $x \in S_1 \cup T_1$ and some constant $c$, then same expression evaluated on $x \in S_0 \cup T_0$ is not strictly on same side of $c$.
\\
If we choose $l_1, l_2, w_1, w_2, c$ such that a) happens, then such weights will not solve training problem as their output of 2-ReLU NN for points $p \in S_1 \cup T_1$ will be at least two distinct numbers which is an undesirable outcome. Specifically, we want $
\relu{w_0+ w_1\relu{l_1(x)} + w_2\relu{l_2(x)}}$ to evaluate to $1$ for all $x \in S_1 \cup T_1$
. Hence $w_1\relu{l_1(x)} +w_2 \relu{l_2(x)}$ must be a constant for all $x \in S_1 \cup T_1$. This requirement is violated in case a).\\
If we choose $l_1, l_2, w_1, w_2, c$ such that b) happens, then we can set $w_0, \theta$ such that $F(x) = \theta \relu{w_1\relu{l_1(x)} + w_2\relu{l_2(x)} + w_0}$, $w_0+c > 0$ and $\theta = \frac{1}{w_0 +c}$. Here we introduced another parameter $\theta > 0$ in the definition of $F$ for sake of convenience of argument but note that $\theta$ can be absorbed in the definition of $l_1$ and $l_2$ to obtain the original neural network function defined \eqref{2relu_nn_def}. Since not all $x \in S_0 \cup T_0$ are strictly on one side, we conclude there exist $x' \in S_0 \cup T_0$ such that $w_1\relu{l_1(x')} + w_2 \relu{l_2(x')} = c' \ge c$ hence $F(x') := \theta \relu{w_1\relu{l_1(x')} + w_2\relu{l_2(x')} + w_0} \ge 1$ which is an undesirable outcome for a point with label $0$.
\\
Since all choices of $l_1, l_2, w_1, w_2, c$ satisfy either a) or b), we conclude that there does not exist weights solving training problem of 2-ReLU NN.

\subsection{Proof of Proposition \ref{prop:simplified}}\label{apx_prop_simplified}
In order to prove Proposition \ref{prop:simplified}, we need to prove one more technical result stated below. Proof of this new proposition is deferred to Appendix \ref{apx_lemma_parallel_lines} but here we state it and proceed with the proof of Proposition \ref{prop:simplified}.
\begin{proposition}\label{lemma_parallel-lines-new}
	Suppose affine functions $a_1, a_2:\Rbb^2\to \Rbb$ and $w_1, w_2 \in \{-1, 1\}$ be such that (i) $a_1(x)$ is a constant for all $x \in \Rbb^2$ or (ii) $a_2(x)$ is a constant for all $x \in \Rbb^2$ or (iii) $w_1a_1(x) + w_2a_2(x)$ is a constant for all $x\in \Rbb^2$, 
	then $a_1,a_2, w_1, w_2$ cannot satisfy hard-sorting of the data points $T_1\cup T_0\cup \{\0\}$ with respect to $T_1\cup \{\0\}$.
\end{proposition}
Now we are ready to prove Proposition \ref{prop:simplified}.

Let $a_1, a_2$ satisfy hard-sorting of $T_1\cup T_0\cup \{\0\}$ with respect to $T_1\cup \{\0\}$. Then due to Remark \ref{subset_prop_hard-sort_remark}, we have that $a_1, a_2$ satisfy hard-sorting of $T_1\cup T_0$ with respect to $T_1$. We will show that any $a_1, a_2$ satisfying the above condition must satisfy the requirement of Proposition \ref{prop:simplified}.

Let us partition the set of points $\Rbb^2$ into four partitions $S_{0,0}, S_{+,0}, S_{0,+}$ and $S_{+,+}$ based on sign of $\relu{a_1(x)}$ and $\relu{a_2(x)}$. Then, we have to show that at least one element in $T_1$ lies in the partition $S_{0,0}$. 

For sake of contradiction, assume that $T_1 \cap S_{0,0} = \emptyset$. Then, using pigeonhole principle, we have that at least one of $S_{+,0}, \ S_{0,+}$ and $S_{+,+}$ must contain three points from the set $T_1$. Note that any three points in the set $T_1$ are not collinear. Moreover, the function $w_1\relu{a_1}+w_2\relu{a_2}$ is affine in all three regions, $S_{+,0}, S_{0,+}$ and $S_{+,+}$ of $\Rbb^2$ and is non-constant in view of Proposition \ref{lemma_parallel-lines-new}. Hence, we cannot satisfy hard-sorting since those three points in $T_1$ will break the requirement in condition \eqref{eq:hard-sort} for hard-sorting. Hence, we obtain a contradiction.

\subsection{Proof of Proposition \ref{lemma_parallel-lines-new}} \label{apx_lemma_parallel_lines}
First observe that if $a_1, a_2:\Rbb^2 \to \Rbb$ satisfy hard-sorting of $T_1 \cup T_0 \cup\{\0\}$ with respect to $T_1 \cup \{\0\}$, then neither of them can be a constant function. In particular, it is straightforward to see that both of them cannot be constant. If only one of them is constant, then data needs to be linearly separable which is not the case for gadget data-points $T_1 \cup T_0 \cup \{\0\}$. Therefore, we will assume that both of them are affine functions with non-zero normal vectors.\\
Note that in view of Remark \ref{rem:w_pm1} and the fact that $w_1a_1(x) + w_2a_2(x) = c$ for all $x \in \Rbb^2$, we may assume that magnitude of the normal to these lines is equal i.e. $\|\nabla a_1\| = \| \nabla a_2 \|\neq 0$. For the sake of this proof, we extend the definition of hard-sorting to include the condition
\[ w_1\relu{a_1(x)} + w_2\relu{a_2(x)} \begin{cases}
=c &\text{for all }x \in T_1\cup\{\0\};\\
>c &\text{for all }x \in T_0,
\end{cases} \]  
along with condition \eqref{eq:hard-sort}. Due to this extended definition and in view of Remark \ref{parity_invariant_remark}, 
we just need to check for case $(w_1, w_2) = (1,1)$ and $(w_1, w_2) = (1,-1)$.  More specifically, $(w_1, w_2) = (-1,-1)$ yields a hard-sorting solution iff there exists a hard-sorting solution for $(w_1,w_2) = (1,1)$. Equivalent argument can be made about the case $(w_1, w_2) = (-1,1)$ and $(w_1,w_2) = (1,-1)$.\\
Then, we have two possible situations here: $a_1, a_2$ satisfy 1) $a_1(x) + a_2(x) = c, \forall \ x \in \Rbb^2$ when normals point in opposite directions and 2) $a_1(x) -a_2(x)=c, \forall x \in \Rbb^2$ when normals point in same direction. 
We will consider both these cases separately and show that expression $w_1\relu{a_1} + w_2\relu{a_2}$, for the choices of $w_1, w_2$ mentioned above, cannot hard-sort the data as required. 
\\
\textbf{Case 1:} Normals point in the opposite directions. Here $w_1 = w_2 = 1$ and we assume $a_1 + a_2 = c$. Suppose $c \ge 0$. Then it can be verified that 
\[\relu{a_1(x)}+\relu{a_2(x)} = \begin{cases}
c &\text{if }c \ge a_1(x) \ge 0\\
a_1(x) &\text{if }a_1(x) \ge c\\
c-a_1(x) &\text{if }a_1(x) \le 0.
\end{cases}\]
By \added{the} 
extended 
hard-sorting requirement, we need all 
points in $T_1 \cup \{\textbf{0}\}$ \replaced{to}{should} be contained in the set $\{x: a_1(x) \in [0,c] \}$ and all points in $T_0$ should not be in this set. Now observe that if $c= 0$, then the set $\{x: a_1(x) = 0\}$ is one dimensional, and therefore cannot contain all the points of $T_1 \cup \{\0\}$. Hence we must have $c > 0$ and all points in $T_1 \cup \{\0\}$ lie inside the region of two parallel lines $a_1(x) = 0$ and $a_1(x) = c$ as $\relu{a_1(x)} +\relu{a_2(x)}$ evaluates to the constant $c$ in this region. It can be seen that this separation of $T_1\cup \{\0\}$ from $T_0$ is impossible to achieve by two parallel lines.

Similarly when $c < 0$, then it can be verified that 
\[\relu{a_1(x)}+\relu{a_2(x)} = \begin{cases}
0 &\text{if }c \le a_1(x) \le 0\\
a_1(x) &\text{if }a_1(x) \ge 0\\
c-a_1(x) &\text{if }a_1(x) \le c
\end{cases}\]
Again, for \added{the extended} hard-sorting, as in the previous case, we need all 
points in $T_1 \cup \{\textbf{0}\}$ should be in set $\{x: a_1(x) \in [c,0] \}$ and all points in $T_0$ should not be in this set which cannot be achieved.\\
%
\textbf{Case 2:} Normals point in the same direction. Then $a_1(x) -a_2(x) = c$. Suppose $c \ge 0$. Then it can be verified that 
\[\relu{a_1(x)}-\relu{a_2(x)} = \begin{cases}
a_1(x) &\text{if }c \ge a_1(x) \ge 0\\
c &\text{if }a_1(x) \ge c\\
0 &\text{if }a_1(x) \le 0
\end{cases}\]
If $c = 0$ then $\relu{a_1(x)} -\relu{a_2(x)} = 0$ for all $x \in \Rbb^2$. So this cannot hard-sort data. Hence for hard-sorting we definitely need $c > 0$. Moreover, we need either 1) $T_1 \cup \{\added{\0}\} \subset \{x: a_1(x) \le 0\}$ and $T_0 \subset \{x: a_1(x) > 0\}$ or 2) $T_1 \cup \{\added{\0}\} \subset \{x: a_1(x) \ge c\}$ and $T_0 \subset \{x: a_1(x) < c\}$. In both cases, \added{we need the points in $T_1 \cup \{\0\}$ must be separable from the points in $ T_0 $ by a line.} This is not possible.\\
Note that when $c < 0$, one can write $a_2 - a_1 = -c$ and write similar functional form for $\relu{a_2}-\relu{a_1}$.\\
Since in both cases, 
we 
were 
un
able to achieve hard-sorting $T_1 \cup T_0 \cup \{\textbf{0}\}$ w.r.t. $T_1 \cup \{\textbf{0}\}$, so we conclude the proof.

\subsection{Proof of Lemma \ref{lemma_parity1}} \label{apx_lemma_parity1}
Proposition \ref{prop:simplified} yields that any hard-sorting $a_1, a_2$ must satisfy $a_1(x)\le0, a_2(x)\le 0$ for at least one $x\in T_1$.\\
Now, suppose sign of $w_1, w_2$ is different. Suppose $w_1 = 1, w_2 = -1$. Since $a_1$ and $a_2$ satisfy hard-sorting of gadget so we have $\relu{a_1(x)} -\relu{a_2(x)}= c, \forall \ x \in T_1$. Due to Proposition \ref{prop:simplified}, we obtain $c = 0$. Then to fulfill hard-sorting condition, we need $\relu{a_1(x)}-\relu{a_2(x)} < 0\ \forall x \in T_0$. (The case for $w_1 = -1, w_2 = 1$ will have same proof with all $a_2$ exchanged by $a_1$ in next 3 lines.)
This implies $a_2(x) > 0$ for all $x \in T_0$. However note that $T_1\subset \text{conv}(T_0)$. So we get a contradiction to the assumption that sign of weights $w_1, w_2$ is different.\\
Now note that if sign of $w_1, w_2$ is same then we cannot set $w_1 = w_2 = 1$ due to requirement \eqref{eq:hard-sort} of hard-sorting. Hence we have that $w_1 =w_2 = -1$. 

\subsection{Proof of Corollary \ref{cor-2-j}}
\label{apx_cor_2j}
The reduction is again from 2-affine separability problem. It involves the same gadget of 18 points in Figure \ref{gadget} and similar labels except that labels need to be extended from $\Rbb$ to $\Rbb^j$. Simply add $j-1$ zeros to original output labels in (\textbf{P}). We call this instance ($\textbf{P}^\prime$).\\
First, we show that if there is a solution to 2-affine separability problem then there is a solution for ($\textbf{P}^\prime$). For the output of the first node, we can use the construction in Lemma \ref{forward_direction} to obtain an exact fit for the first node. For rest $j-1$ nodes, the output is $0$ for all data-points so we can easily extend the solution to obtain an exact fit for all nodes. In particular, for $k \in [j]$, every $k$-th node in the second layer is connected to two nodes in the first layer by distinct edges whose weights are parameterized by $ w_{k,1}, w_{k,2}$ and bias weight $w_{k,0}$. We can set $w_{k,1}= w_{k,2} = -1$ and $w_{k,0} = 0$ for all $k \in [j] \setminus \{1\}$. The output at $k$-th node can be written as $\relu{w_{k,1}\relu{a_1(x)} +w_{k,2}\relu{a_2(x)} + w_{k,0}}$. In view of the above values of $w_{k,0}, w_{k,1}$ and $w_{k,2}$ for $k \in [j] \setminus \{1\}$, we note that $w_{k,1}\relu{a_1(x)} +w_{k,2}\relu{a_2(x)} + w_{k,0} \le 0$ for all $k \in [j] \setminus \{1\}$. This yields the output $0$ at all nodes $k \in [j] \setminus\{1\}$, irrespective of the affine functions $a_1, a_2$ in the first layer. Hence all nodes are satisfied to global optimality which completes the forward direction of the reduction.

Reverse direction follows immediately from the construction of the gadget and Lemma \ref{reverse_direction}. To see this, note that solution for ($\textbf{P}^\prime$) implies that first node is satisfied to global optimality. Ignoring the rest $j-1$ nodes, we see that this solution is also a solution for (\textbf{P}). At this point, invoking Lemma \ref{reverse_direction} yields the solution corresponding to 2-hyperplane separability problem. Hence, we conclude that training problem of $(2,j)$-ReLU NN is NP-hard.

\end{appendices}


\bibliography{deeplearning-complexity}

\end{document}